\documentclass[sigconf]{acmart}
\AtBeginDocument{%
  \providecommand\BibTeX{{%
    \normalfont B\kern-0.5em{\scshape i\kern-0.25em b}\kern-0.8em\TeX}}}

\copyrightyear{2022}
\acmYear{2022}
\setcopyright{acmcopyright}
\acmConference[WSDM '22] {Proceedings of the Fifteenth ACM International Conference on Web Search and Data Mining}{February 21--25, 2022}{Tempe, AZ, USA.}
\acmBooktitle{Proceedings of the Fifteenth ACM International Conference on Web Search and Data Mining (WSDM '22), February 21--25, 2022, Tempe, AZ, USA}
\acmPrice{15.00}
\acmISBN{978-1-4503-9132-0/22/02}
\acmDOI{10.1145/3488560.3498373}

\usepackage{graphicx}
\usepackage{subfigure}
\usepackage{multirow}
\usepackage[ruled]{algorithm2e} 
\usepackage{bbm}
\usepackage{amsthm}

\begin{document}
\fancyhead{} 

\title{A Cooperative-Competitive Multi-Agent Framework for Auto-bidding in Online Advertising}


\author{Chao Wen$^{1,3*}$, Miao Xu$^{3}$, Zhilin Zhang$^{3}$, Zhenzhe Zheng$^{2\dagger}$, \\ Yuhui Wang$^{1}$, Xiangyu Liu$^{3}$, Yu Rong$^{3}$  Dong Xie$^{3}$, Xiaoyang Tan$^{1}$, \\ Chuan Yu$^{3}$, Jian Xu$^{3}$, Fan Wu$^{2}$, Guihai Chen$^{2}$, Xiaoqiang Zhu$^{3}$ and Bo Zheng$^{3}$}
\affiliation{%
  \institution{
    $^1$MIIT Key Laboratory of Pattern Analysis and Machine Intelligence \country{China}\\ 
    \and $^2$Shanghai Jiao Tong University \country{China}, $^3$Alibaba Group \country{China}
    \and $^1$\{chaowen, y.wang, x.tan\}@nuaa.edu.cn, $^2$\{zhengzhenzhe@, fwu@cs., gchen@cs.\}sjtu.edu.cn,
    \and $^3$\{xumiao.xm, zhangzhilin.pt, qilin.lxy, homer.ry robber.xd, yuchuan.yc, xiyu.xj, xiaoqiang.zxq, bozheng\}@alibaba-inc.com
  }
}

\thanks{This work was supported in part by Science and Technology Innovation 2030 – ``New Generation Artificial Intelligence'' Major Project No. 2018AAA0100905, China NSF grant No. 61902248, 61976115, in part by Shanghai Science and Technology Fund 20PJ1407900, in part by Alibaba Group through Alibaba Innovation Research Program and Alibaba Research Intern Program. The opinions, findings, conclusions, and recommendations expressed in this paper are those of the authors and do not necessarily reflect the views of the funding agencies or the government.}
\thanks{$^*$Work done during an internship at Alibaba Group.}
\thanks{$\dagger$Corresponding author.}

\renewcommand{\shortauthors}{Wen and Xu, et al.}

\begin{abstract}
In online advertising, auto-bidding has become an essential tool for advertisers to optimize their preferred ad performance metrics by simply expressing high-level campaign objectives and constraints. Previous works designed auto-bidding tools from the view of single-agent, without modeling the mutual influence between agents.  In this paper, we instead consider this problem from  a distributed multi-agent perspective, and propose a general \underline{M}ulti-\underline{A}gent reinforcement learning framework for \underline{A}uto-\underline{B}idding, namely MAAB, to learn the auto-bidding strategies. 
First, we investigate the competition and cooperation relation among auto-bidding agents, and propose a temperature-regularized credit assignment to establish a mixed cooperative-competitive paradigm. By carefully making a competition and cooperation trade-off among agents, we can reach an equilibrium state that guarantees not only individual advertiser's utility but also the system performance (i.e., social welfare). Second, to avoid the potential collusion behaviors of bidding low prices underlying the cooperation, we further propose bar agents to set a personalized bidding bar for each agent, and then alleviate the revenue degradation due to the cooperation. Third, to deploy MAAB in the large-scale advertising system with millions of advertisers, we propose a mean-field approach. By grouping advertisers with the same objective as a mean auto-bidding agent, the interactions among the large-scale advertisers are greatly simplified, making it practical to train MAAB efficiently. Extensive experiments on the offline industrial dataset and Alibaba advertising platform demonstrate that our approach outperforms several baseline methods in terms of social welfare and revenue.
\end{abstract}



\maketitle

\section{Introduction}

In recent years, online advertising has become ubiquitous in modern advertising markets, and generates hundreds of billions of dollars every year~\cite{emarketer2015worldwide}. 
Online advertising allows advertisers to increase the exposure opportunities of their products to potential audiences. 
Traditionally, advertisers need to manually adjust a bid in each ad auction to optimize the overall ad campaign performance. 
However, this fine-grained bidding process requires domain knowledge and comprehensive information about advertising environments~\cite{yang2019aiads}. 
To reduce the burden on bid optimization for advertisers, online platforms have deployed various types of auto-bidding services, such as Google's AdWords campaign management tool~\cite{google2021adwords} and Facebook's automated bidding services~\cite{facebookautobidding2021}.
These services allow advertisers to simply express high-level campaign objectives and constraints, and the auto-bidding agents would calculate the bids for each auction on behalf of advertisers to optimize their preferred ad performance metrics~\cite{google2021autobidding}.

\begin{figure}[t]
    \centering
    \includegraphics[width=0.9\linewidth]{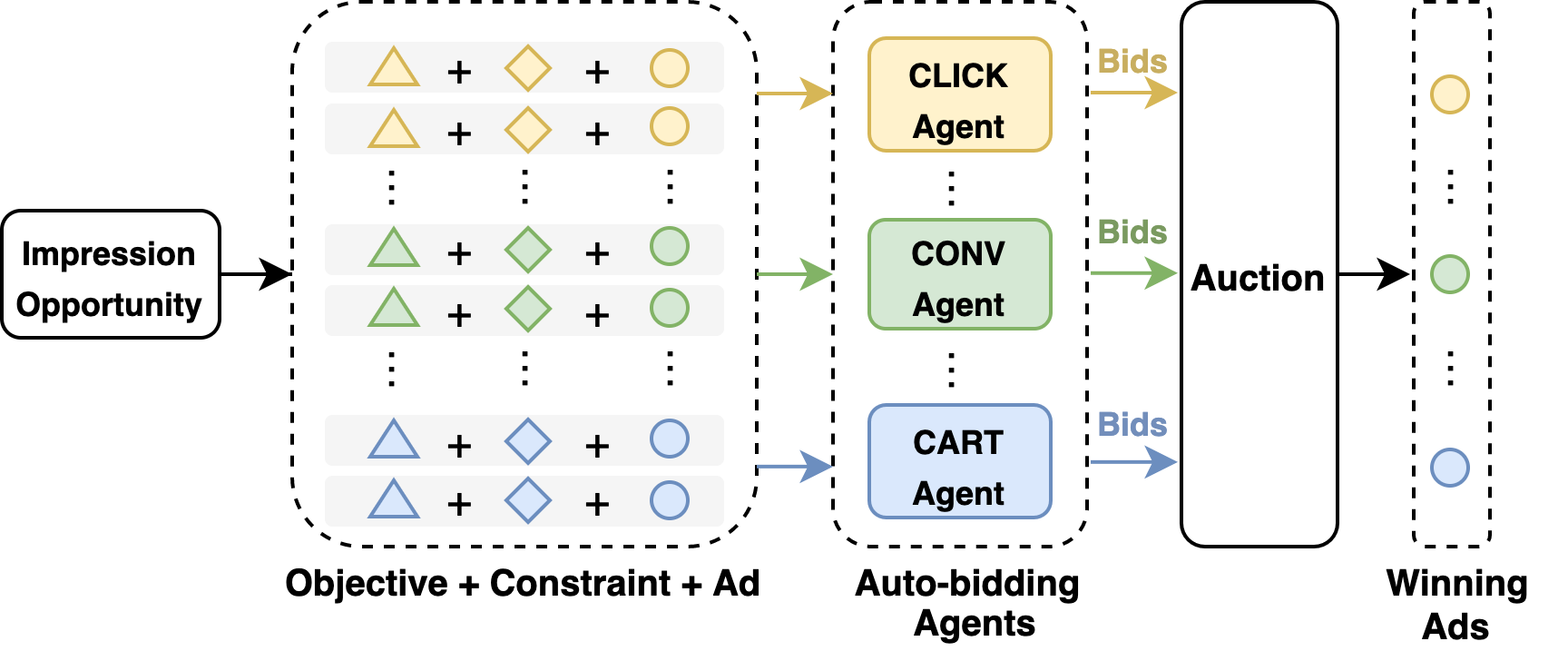}
    \caption{An Overview of Auto-bidding Services.}
    \label{fig.autobidding}
\end{figure}

We illustrate the procedure of auto-bidding services for online advertising in Figure \ref{fig.autobidding}. The advertising platform designs auto-bidding agents, each of which learns a bidding strategy for the advertiser to optimize her desired objective under certain constraints. For example, in Alibaba advertising platform, there exist three typical types of auto-bidding agents: CLICK Agent, CONV Agent, and CART Agent, which optimize the number of clicks, purchase conversions and add-to-carts under budget constraints, respectively.
These auto-bidding agents compete with each other by bidding for each ad display opportunity (\emph{impression opportunity}).
After receiving an impression opportunity, the platform launches an ad auction for the participating auto-bidding agents, each of which places a bid according to the advertiser's objective and constraint, as well as the predicted value (\emph{e.g.}, pCTR, pCVR or pCART) of this impression opportunity.
After collecting bids from all the agents, 
the auction mechanism determines the winning advertisers and the corresponding payments. 


To learn the bidding strategies for auto-bidding agents, a natural option, also widely adopted in real-time-bidding (RTB) literature~\cite{aggarwal2019autobidding,wu2018budget}, is to 
solve an isolated optimization problem for each auto-bidding agent, and the effects of other agents are implicitly encoded within the auction environment. However, this formulation ignores the fact that the ad auction mechanism is inherently a  distributed multi-agent system in nature -- the outcome of an ad  auction depends heavily on the bidding behaviors of all the involved auto-bidding agents~\cite{jin2018real}. 
Without appropriate coordination, the auto-bidding agents would lead to an anarchy state with significantly degraded system performance, such as
the \emph{tragedy of the commons}~\cite{hardin2009tragedy}. 
Thus, we leverage a distributed multi-agent framework associated with a carefully designed coordinated scheme to steer the behaviors of auto-bidding agents towards an equilibrium state with good system performance.

There are still several challenges in designing the bidding strategies for auto-bidding agents under this framework~\cite{google2021autobidding,yang2019aiads}. 
First, the complex competition and cooperation relations among auto-bidding agents hinder the platform to jointly optimize the individual agent's utility and the overall system performance.
%
On the one side, 
in a fully competitive multi-agent paradigm, each individual advertiser's utility can be extremely optimized. 
For example, the auto-bidding agents with adequate budgets would bid aggressively and dominate a large number of ad auctions, for their own interests. However, the winning advertisers may contribute low social welfare, due to an inefficient ad allocation in a system-level view. 
On the other hand, a fully cooperative multi-agent scheme can resolve this issue as all the auto-bidding agents aim at maximizing the social welfare together, as in the centralized optimization. However, this method may sacrifice some advertisers' utilities for a large social welfare, which is not healthy for the long-term prosperity of online advertising.  We need to resolve the conflict between individual advertiser's utility and the system performance (social welfare). 
A more proper way is to establish a mixed cooperative-competitive (MCC) framework, enabling the platform to make a flexible trade-off between individual utility and the system performance. 
Existing approaches to achieve this are through either manually modifying reward functions~\cite{tampuu2017multiagent} or changing environment-related factors~\cite{Leibo2017,lowe2017multi}. 
However, there is still no specific reward function for auto-bidding agents in the context of ad auctions, and
the environmental factors are only controllable in a simulator 
but not in the real-world online setting. 

Second, the cooperation in the MCC framework would inevitably reduce the platform's revenue. 
This is because cooperative auto-bidding agents may collude with each other to bid low prices~\cite{marshall2007bidder}.
We can leverage the reserve price~\cite{yuan2014empirical,mohri2014learning}, a classical method derived from the optimal auction theory~\cite{myerson1981optimal}, to boost the revenue of auction.
However, the optimal reserve price is usually calculated in the non-cooperative auction setting with simple  bidders/agents, and how to set the optimal reserve prices in the MCC framework that guarantees platform's revenue without reducing social welfare is still unclear. 


The last but not least, the industrial deployment of a MCC multi-agent framework for auto-bidding services is quite challenging. In a practical advertising platform, there are millions of advertisers competing for billions of impressions every day, thus it is difficult to formulate each advertiser as a single agent and train millions of agents simultaneously with limited computational resources and time. Besides, the sparsity of the rewards for each agent cannot properly direct the policy learning due to the limited impression opportunities with respective to the huge number of advertisers. 

By jointly considering the above design challenges, we propose a cooperative-competitive multi-agent reinforcement learning (MARL) framework for auto-bidding in online advertising, namely MAAB. 
First, to handle the trade-off between the cooperative and competitive relation among auto-bidding agents,
we propose a temperature-regularized credit assignment scheme, which distributes the total reward from the auction among the participating agents in proportion to the weights captured by a softmax function. 
The temperature parameter in the softmax function could work as a controller to adjust the degree of the trade-off.
Second, to reduce the loss of revenue, we design bar agents for learning a personalized bidding bar for each auto-bidding agent. 
Intuitively, the bar agent's goal is to increase the bidding bar for a high revenue for the platform, while the auto-bidding agent has an opposite goal --- reducing the bidding bar for a low payment. The bar agent's training is implemented through an adversarial interaction with the training for auto-bidding agents until reaching some equilibrium point.
Third, we propose a mean-field style approach~\cite{lasry2007mean,yang2018mean} to tackle the 
challenge from the industrial large-scale multi-agent system. 
By grouping the advertisers with the same objective as a mean auto-bidding agent, the complex interactions among millions of advertisers are significantly reduced to the interactions within a limited 
number of super auto-bidding agents, making it practical to train and deploy the large-scale multi-agent system for auto-bidding services. 
Extensive offline and online experiments demonstrate the effectiveness of our MAAB in terms of both social welfare and platform's revenue.\footnote{Code is available at \href{https://github.com/chaovven/maab}{https://github.com/chaovven/maab}.} 

The contributions in this paper can be summarized as follows:
\begin{enumerate}
\item We investigated auto-bidding problem from the perspective of multi-agent systems, and proposed temperature-regularized credit assignment to properly model the mixed cooperative-competitive relation among auto-bidding agents.
\item We designed the bar agents to solve the multiple objectives in the cooperative-competitive auto-bidding framework, where platform's revenue can be improved without reducing social welfare by adaptively setting the bidding bars.
\item We proposed a mean-field style approach for handling the 
scalability of large-scale multi-agent system for auto-bidding. The evaluation results show the effectiveness of our method in terms of both social welfare and platform's revenue.

\end{enumerate}

\section{Preliminaries}
\label{sec:preliminary}

\subsection{Auto-bidding Model}
We consider the budget-constrained auto-bidding services, where advertisers hope to maximize the total value of winning impressions under a budget constraint~\cite{wu2018budget}.
The auto-bidding process is described as follows.
During a period (e.g., one day) with $T$ impression opportunities arriving sequentially, the auto-bidding agent gives the bid $b_i^t$ on behalf of advertiser $i$ for the impression opportunity arriving at timestep $t$.
If $b_i^t$ is the highest bid in the auction, advertiser $i$ displays her ad, obtains the impression value $v_i^t$, and makes the payment $p^t$, which is the second-highest bid in the second price auction.
This  bidding process terminates if the total payment reaches the budget limit or there are no impression opportunities left. 
The goal of auto-bidding agent that bids on behalf of advertiser $i$ is to maximize the total value of winning impressions under the budget constraint:
\begin{align}
     & \max \sum_{t=1}^T v_i^t \times x_i^t \quad     \\
     & s.t. \sum_{t=1}^T p^t \times  x_i^t \leq B_i,
\end{align}
where $x_i^t \in \{0, 1\}$ denotes whether advertiser $i$ wins the impression $t$. In the following discussion, we also use $i$ to denote the auto-bidding agent of the advertiser $i$.

\subsection{Markov Games}
\label{subsec:markov_games}
The previous subsection describes auto-bidding services from the view of single-agent, we next use Markov Game ~\cite{littman1994markov} to model the interaction among multiple auto-bidding agents. We briefly recall the notations of Markov game. 
A partially observable Markov Game of $n$ agents is a tuple consisting of 
$<\mathcal{S}, P, \{Z_i, \mathcal{O}_i, \mathcal{A}_i, r_i\}_{i=1}^n, \gamma>$. 
We use $s \in \mathcal{S}$ to denote the state of the environment. 
At each timestep, each agent takes an action $a_i=\pi_i(o_i) \in \mathcal{A}_i$ according to its policy $\pi_i$ and an observation $o_i$, drawing from the observation function $Z_i:\mathcal{S} \rightarrow O_i$.
After all the agents taking the joint action $\mathbf{a} = (a_1, \cdots, a_n)$, each agent $i$ obtains a scalar reward $r_i: \mathcal{S} \times \mathcal{A}_1 \times \cdots \times \mathcal{A}_n \rightarrow \mathbb{R}$, and environment transits to the next state $s'$ according to the transition function $P(s'|s,\mathbf{a}): \mathcal{S} \times \mathcal{A}_1 \times \cdots \times \mathcal{A}_n \rightarrow \mathcal{S}$.  $\gamma \in (0,1]$ is a discount factor for future rewards.
 Each agent aims to learn its policy in order to maximize its expected accumulated reward $R_i = \mathbb{E}[\sum_{t=1}^T \gamma^{t-1} r_i^t]$ over the period.

We describe the above Markov Game in the context of auto-bidding services. 
At each timestep $t$, an auto-bidding agent $i$ decides the bid price $b_i^t=\pi_i(o_i^t)$ according to the received observation $o_i^t=(B_i^t, v_i^t, ts_i^t)$, where $B_i^t$ is the remaining budget (the bid price is forced to be 0 when the remaining budget is 0), $v_i^t$ is the impression value and $ts_i^t = T-t$ is the timesteps left. 
The reward for the winning auto-bidding agent is $r_i^t=v_i^t$ and the payment $p^t$ is determined by the second-highest bid. 
After that, the environment transits to the next state, and the new observation for each agent is $o_{i}^{t+1}=(B_i^t-p^t\times x_i^t, v_i^{t+1}, ts_i^{t}-1)$. 
The objective of auto-bidding agent is to maximize its expected total value of winning impressions: $\max_{\pi_i} \mathbb{E}[\sum_{t=1}^T \gamma^{t-1} r_i^t]$.

\subsection{Independent Learner}
The classical multi-agent reinforcement learning (MARL) algorithm for solving Markov Games is to directly learn decentralized value functions or policies simultaneously \cite{tan1993multi,tampuu2017multiagent}. 
Independent $Q$-learning \cite{tan1993multi,tampuu2017multiagent} trains individual $Q$ function for each agent, with all agents sharing the same environment.
We refer to independent $Q$-learning as \emph{independent learner (IL)}. We represent each agent's action-value function $Q_i(o_i, b_i)$ that conditions on its observations and bids with a deep neural network parameterized by $\theta_i$'s. 
A replay buffer $\mathcal{D}$ stores the transition tuple of all agents $\{(o_i, b_i, r_i, o_{i}^{\prime})\}_{i=1}^n$, where $o_{i}^{\prime}$ is observed by agent $i$ after submitting the bid $b_i$ based on its observation $o_i$ and receiving reward $r_i$.
The training process of $Q_i(o_i,b_i)$ is similar to DQN \cite{mnih2015human}, where the parameters $\theta_i$'s are learned by sampling batches of transitions from the replay memory, and minimizing the following loss function:
\begin{equation}
    \mathcal{L}(\theta_i)=\mathbb{E}_{(o_i,b_i,r_i, o_{i}^{\prime}) \sim \mathcal{D}} \left[(y_i-Q_i(o_i, b_i; \theta_i))^{2}\right],
\end{equation}
where the target $y_{i}=r_{i}^{\text{train}}+\gamma \max_{b_{i}^{\prime}} Q(o_{i}^{\prime}, b_{i}^{\prime} ; \hat{\theta}_i)$, and the temporary parameters $\hat{\theta}_i$’s are periodically updated to the new  $\theta_i$ after a fixed number of timesteps.
We use $r^{\text{train}}_i$ to denote the reward that is used to learn $\theta_i$.
We can tune this training reward $r^{\text{train}}_i$ to capture different types of relations among the agents: 1) we can set $r^{\text{train}}_i$ as the \emph{individual reward} $r_i$ of each agent to model the competition among agents, and the resulting IL is called competitive IL (\textbf{CM-IL}); and 2) we can also set $r^{\text{train}}_i$ as the \emph{total reward} $r^{\text{tot}}=\sum_{i=1}^n r_i$ of all agents, the social welfare of the multi-agent system, under which all the agents jointly optimize the performance of the whole system. 
We refer to the corresponding IL as cooperative IL (\textbf{CO-IL}). 


We also interpret the cooperative and competitive relations of agents in the context of ad auctions. Consider a one-round auction with two agents, the impression values of which are $v_1$ and $v_2$, respectively (assume $v_1>v_2$ without loss of generality). 
The relation is \emph{cooperative} if their bids satisfy $b_1 > b_2$, otherwise is \emph {competitive}. This interpretation accords with the intuition that cooperation contributes to better social welfare. 

\section{Behaviors of Independent Learners}
\label{sec:emergent_beha}

In this section, we demonstrate the interaction results of CM-IL and CO-IL agents in the context of auto-bidding, through numerical experiments in a simplified two-agent bidding environment. We find that for CM-IL agents, the phenomena of oligarch emerges, leading to fierce competition and worse social welfare, while CO-IL agents  achieve better social welfare, however, at the cost of reducing the platform's revenue.

\begin{figure}[t!]  
    \centering
    \subfigure[\footnotesize{Agent 1's winning values for CM-IL}]{
        \label{fig.simp_co0_ader1}
        \includegraphics[width=0.23\textwidth]{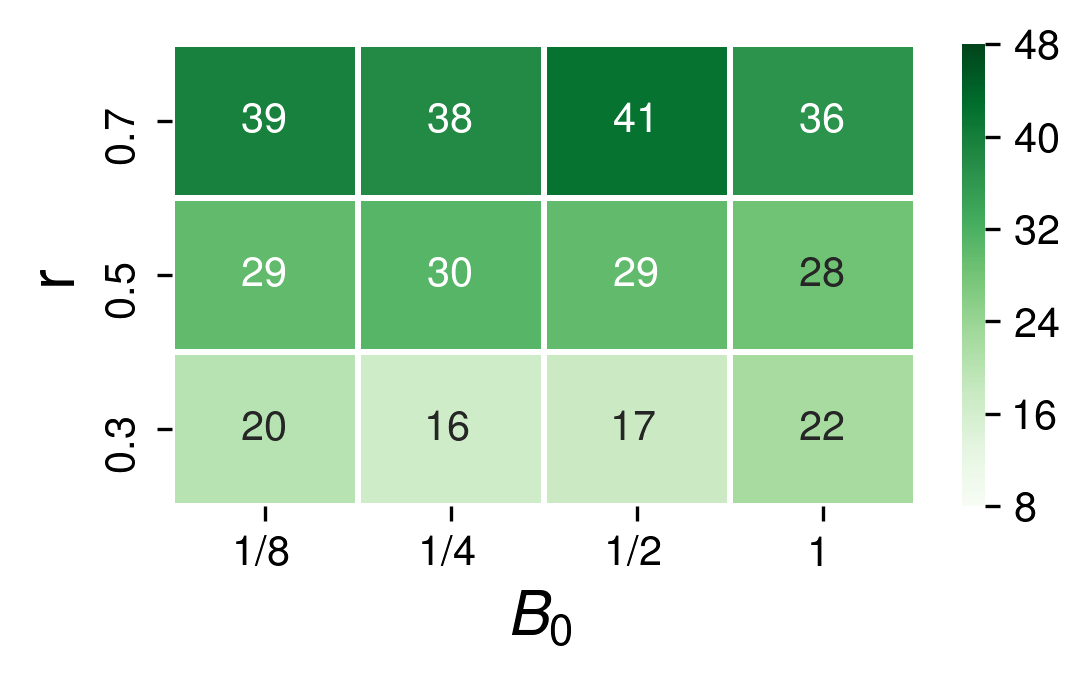}}
    \subfigure[\footnotesize{Agent 1's winning values for CO-IL}]{
        \label{fig.simp_co100_ader1}
        \includegraphics[width=0.23\textwidth]{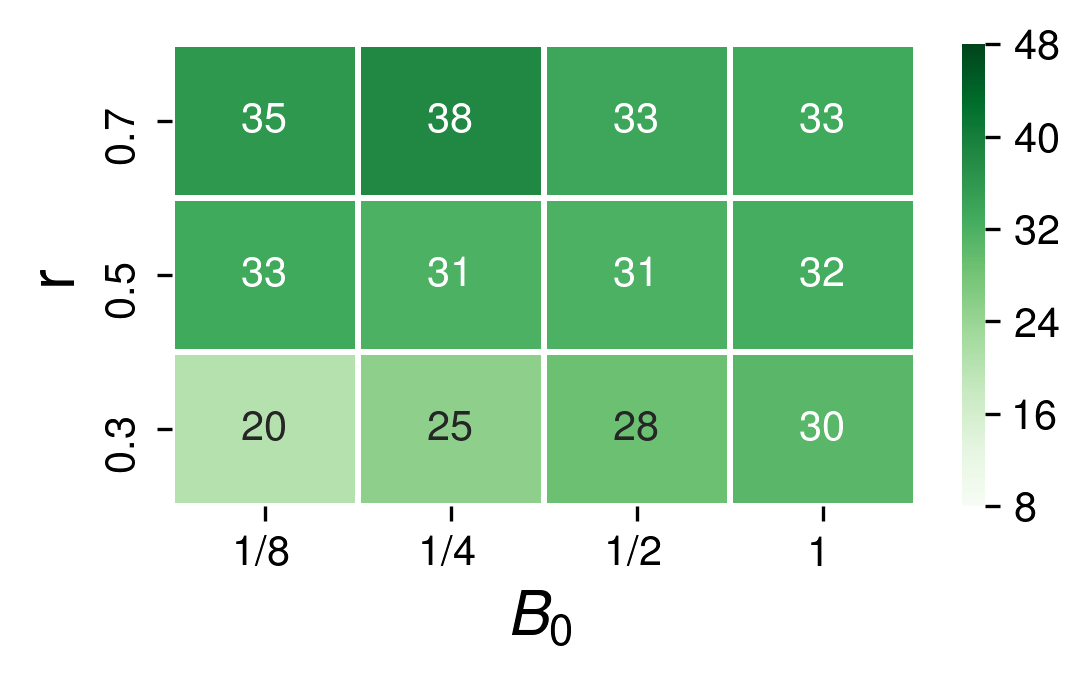}}
    \subfigure[Social welfare for CM-IL]{
        \label{fig.simp_co0_sw}
        \includegraphics[width=0.23\textwidth]{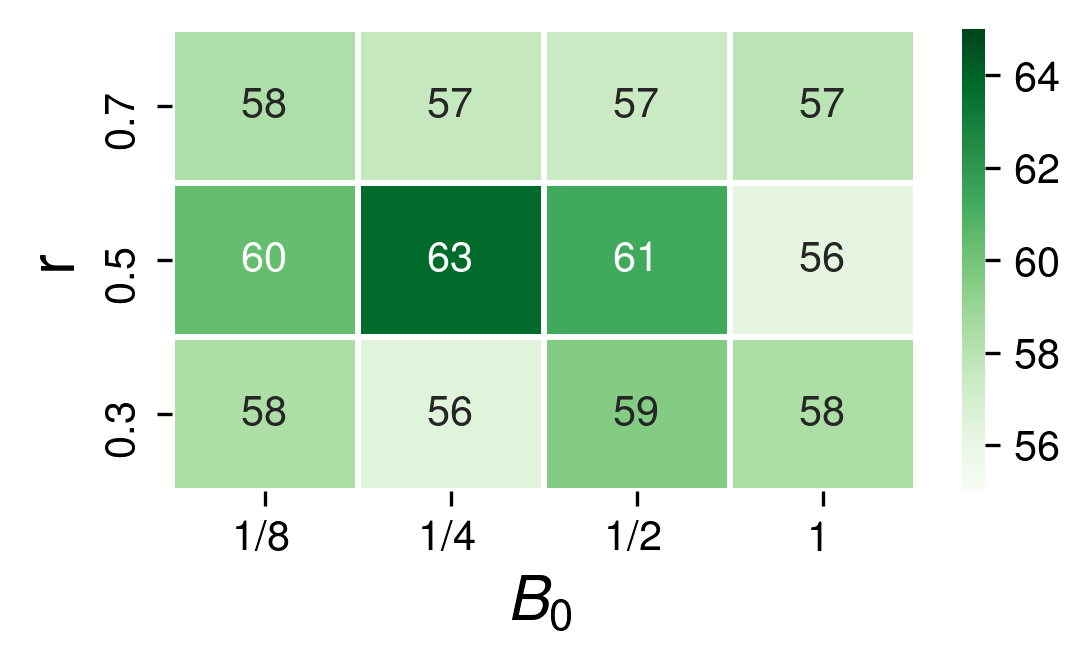}}
    \subfigure[Social welfare for CO-IL]{
        \label{fig.simp_co100_sw}
        \includegraphics[width=0.23\textwidth]{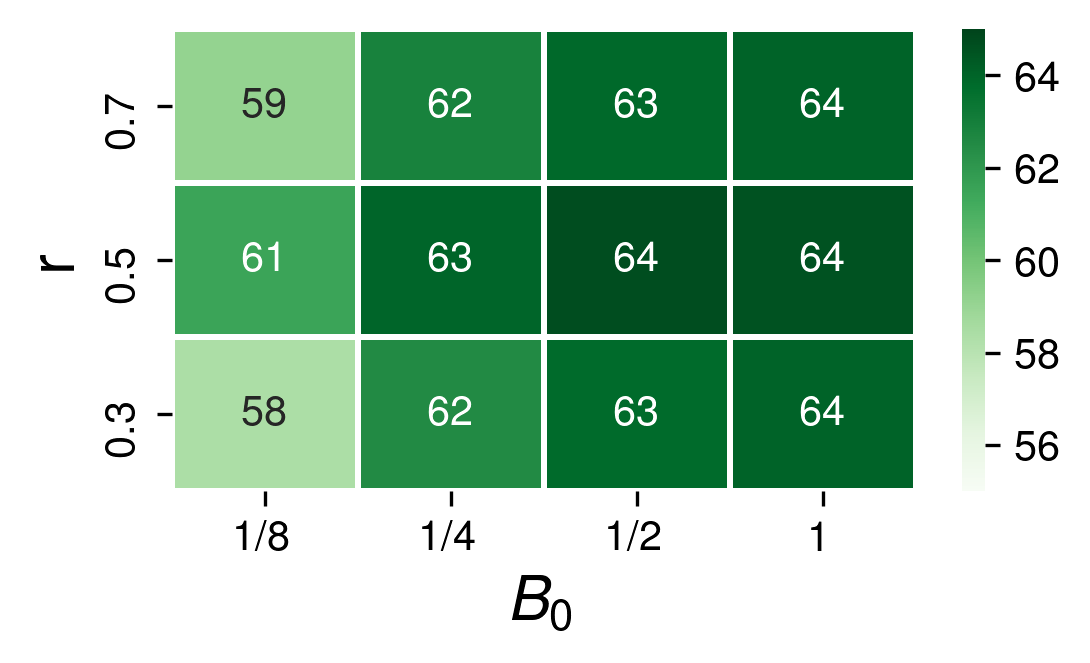}}
    \subfigure[Revenue for CM-IL]{
        \label{fig.simp_co0_pr}
        \includegraphics[width=0.23\textwidth]{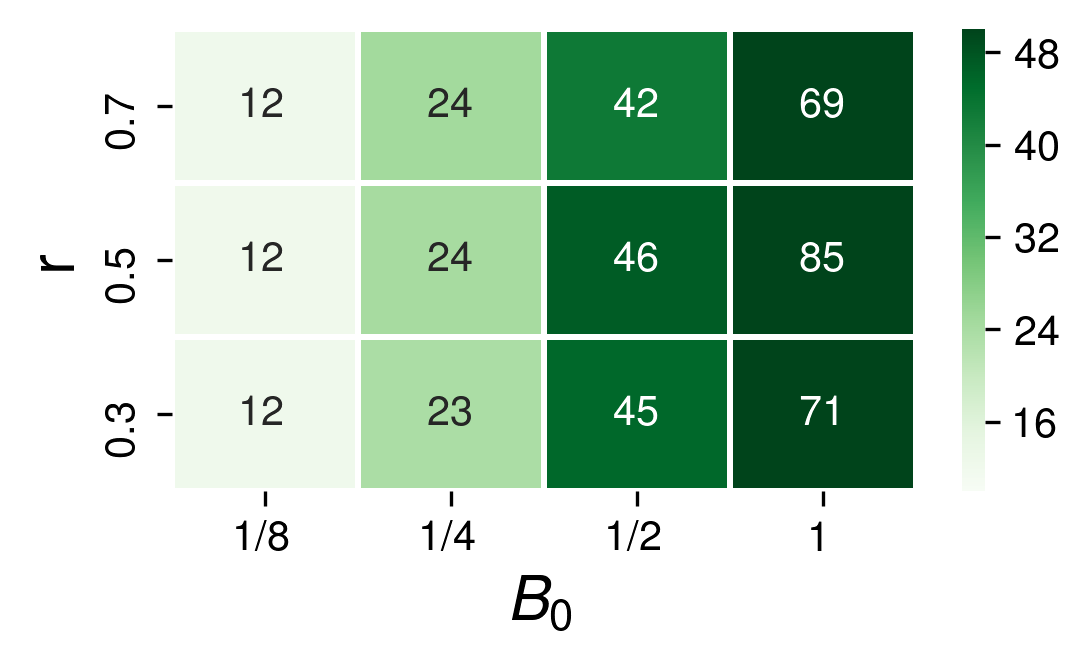}}
    \subfigure[Revenue for CO-IL]{
        \label{fig.simp_co100_pr}
        \includegraphics[width=0.23\textwidth]{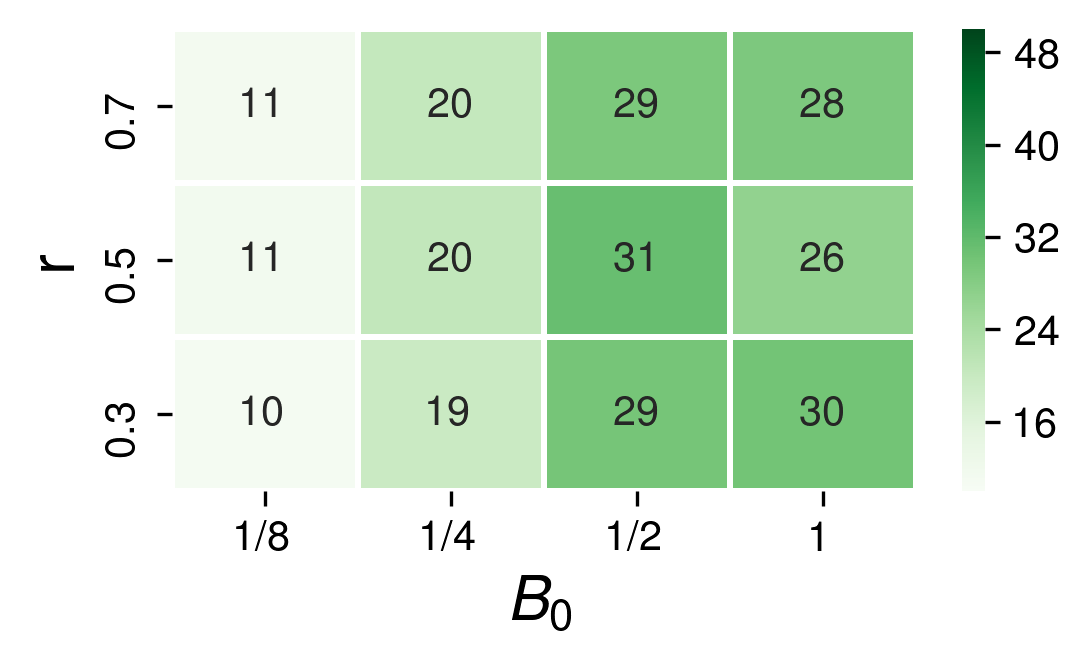}}
    \caption{Converged performance of CM-IL and CO-IL.}
    \label{fig.simp_moti}
\end{figure}

We devise an auction environment with two bidding agents, 
which can be implemented with either the paradigm of competitive CM-IL or cooperative CO-IL.     
We vary the experimental settings with different total budget $B_0$ (normalized) in each episode and the budget ratio $r$, which 
controls the percentage of the total budget to agent $1$. Thus the budgets for the two agents  are $B_1= B_0 \times r$ and $B_2=B_0 \times (1-r)$, respectively. Both agents' impression values are sampled from the normal distribution with mean $0.5$ and variance $1$. We train CM-IL and CO-IL  for $50k$ episodes, following the training principles discussed above.
We examine their final performance in terms of three metrics: 1) agent 1's total winning values, 2) social welfare of the auction and 3) revenue of the auction.
Social welfare is the sum of two agents' winning values, and revenue is the total payment calculated with Generalized Second Price (GSP) mechanism~\cite{edelman2007internet}. 
We omit the total winning values for agent $2$, as it can be derived by subtracting agent 1's winning value from the corresponding social welfare, e.g., we can derive that for $r=0.7$, agent 2's winning values are (19, 19, 16, 21) under different values of $B_0$.
We plot the experimental results in Figure \ref{fig.simp_moti}. Each number in the cell denotes agent 1's obtained winning values, social welfare or revenue with specific experimental parameters (i.e., $B_0$ and $r$). 

Figure \ref{fig.simp_co0_ader1} shows the total winning values of agent 1. Given an unbalanced budget setting $r=0.7$, we observe that agent 1's total winning  values $(39, 38, 41, 36)$ under different total budget $B_0$, are significantly larger than those of agent 2 $(19, 19, 16, 21)$. 
This indicates that agent 1 dominates most impressions by bidding aggressively, leading to a phenomenon of \emph{oligarch}~\cite{davidson1986long}. 
The emergent oligarch worsens the social welfare. As shown in Figure \ref{fig.simp_co0_sw} and Figure \ref{fig.simp_co100_sw}, CM-IL achieves the less social welfare than CO-IL, especially in enough budget settings, e.g., when $B_0=1$, CO-IL's social welfare $(64, 64, 64)$ are larger than those of CM-IL $(57, 56, 58)$. 


A proper cooperation improves the social welfare by preventing the emergent oligarch. This can be seen by comparing Figure \ref{fig.simp_co0_ader1} with \ref{fig.simp_co100_ader1}: agent 1's total winning values decrease from $(39, 38, 41, 36)$ to $(35, 38, 33, 33)$ with more budget ($r=0.7$), and increase from $(20, 16, 17, 22)$ to $(20, 25, 28, 30)$ with  less budget ($r=0.3$). 
This suggests that CO-IL assigns the impression opportunity mostly according to the value of impression instead of the budget, which turns out to be a better equilibrium in terms of social welfare in these settings.

However, the above assign-by-value rule taken by CO-IL agents may sacrifice some advertisers' profits to achieve better social welfare.
Besides, Figure \ref{fig.simp_co0_pr} and Figure \ref{fig.simp_co100_pr} also show that cooperation also lowers the platform's revenue. For example, when $B_0=1$, the revenues achieved by CO-IL agents (28, 26, 30) are significantly lower than those of CM-IL agents (69, 85, 71). This is because CO-IL agents have learned a collusion behavior of bidding low prices to reserve more budget.

The opposite behaviors of the competitive approach and its cooperative counterpart come to two extremes: oligarch arises in the former approach under unbalanced budget settings, which causes worse social welfare. In contrast, the cooperative approach achieves better social welfare, however, at the cost of reducing the platform's revenue and may potentially sacrifice some advertisers' profits for achieving better social welfare.

\section{Methods}
In this section, we present our multi-agent framework MAAB, aiming to achieve good social welfare while guaranteeing the platform's revenue. The architecture of MAAB is shown in Figure \ref{fig.model}, and the training procedure is provided in Appendix \ref{app:algo}.


\subsection{Mixing Cooperation and Competition}
Motivated by the extreme behaviors of CM-IL and CO-IL, we propose a reward assignment scheme, called temperature regularized credit assignment (TRCA), to establish a mixed cooperative-competitive (MCC) \cite{tampuu2017multiagent} relation among agents.

\begin{figure}[t]
    \centering
    \includegraphics[width=0.95\linewidth]{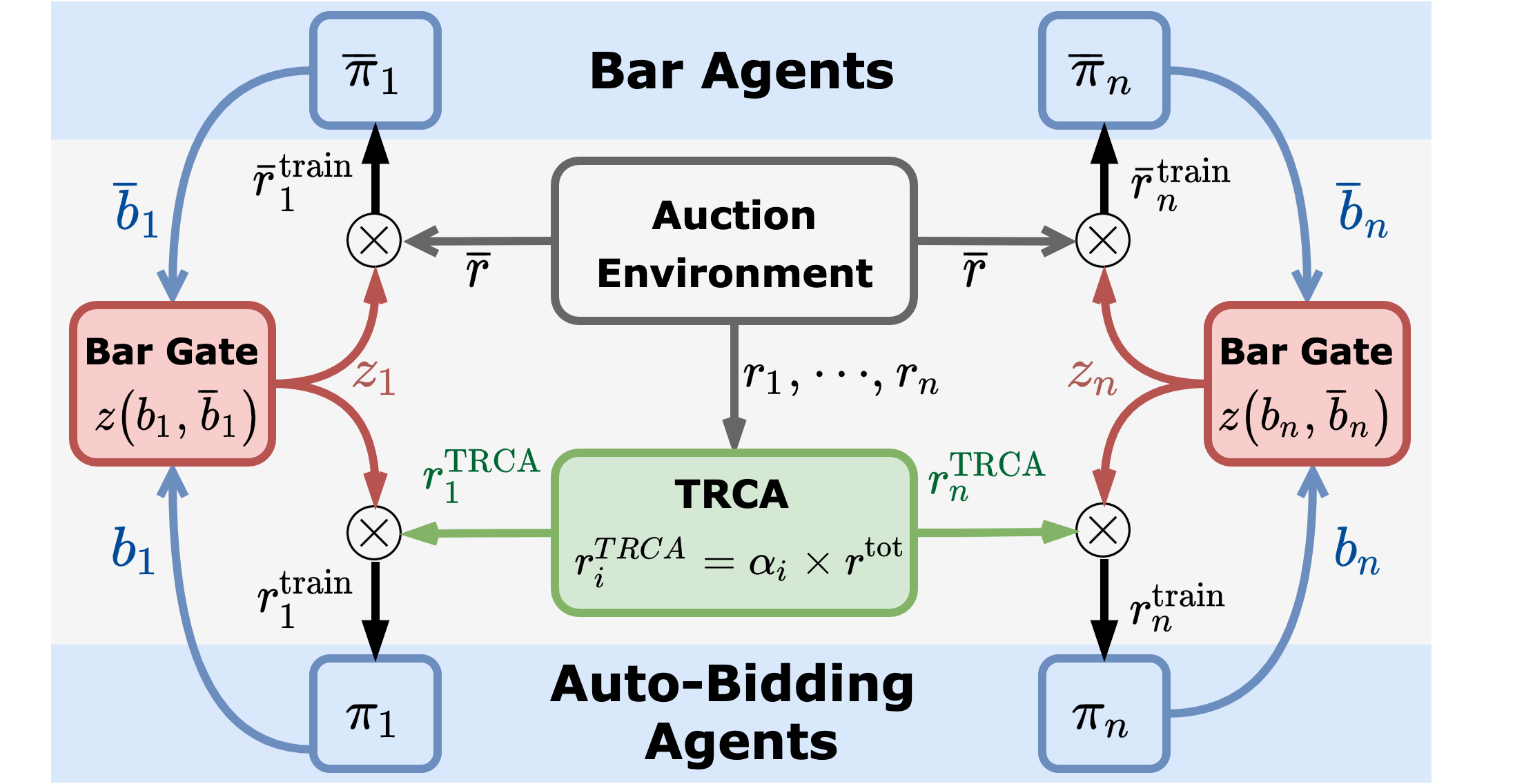}
    \caption{The architecture of MAAB.}
    \label{fig.model}
\end{figure}

The main idea is to set a parameter $\alpha_i$ weighting each agent's contribution to the total reward. The reward for each agent $i$ can be written as
\begin{equation}
    r_i^{\text{TRCA}} = \alpha_i \times  r^{\text{tot}},
    \label{eq:trca}
\end{equation}
where $\alpha_i =\frac{\text{exp}\{b_i/\tau\}}{\sum_{j=1}^n \text{exp}\{b_j/\tau\}}$ is a softmax-style weighting parameter that satisfies $\alpha_i \in [0,1]$ and $\sum_{i=1}^n \alpha_i=1$. The $\tau$ is a temperature parameter balancing the extent of competition and cooperation.
The intuition behind Eq. \ref{eq:trca} is that the highest bid dominates the value of $r^{\text{tot}}$, while a lower bid  has little influence on it. Thus, the assigned reward $r^{\text{TRCA}}_i$ for each agent should increase with her bid, which is achieved by setting $\alpha_i$ in proportion to $b_i$ by a softmax function. 

Our designed reward function also encodes both social welfare and revenue. The social welfare, measured by $r^{tot}$, appears in Eq. \ref{eq:trca} directly, while the revenue contributed by each agent is implicitly expressed by her bid, as higher bids collected from all agents usually leads to higher revenue under the GSP auction.


In addition, the temperature parameter $\tau$ in Eq. \ref{eq:trca} enables the co-existence of competition and cooperation, and works as a convenient tool to make a trade-off between these two relations. To illustrate this, we theoretically investigate how the value of $\tau$ affects the bidding behaviors and the relations between agents in the simplified case of two agents. 


\begin{theorem}
        Consider a two-agent bidding case with impression values satisfying $v_1 > v_2$ in one round auction. Let $b_1, b_2 \in [b_{min}, b_{max}]$ denote two agents' bids, where $b_{min}$ and $b_{max}$ are the minimum and maximum allowable bids respectively. If $ v_1 \geq 2 v_2$ or

        \begin{equation}
			\tau \geq \frac{\log \left( 2v_2/v_1-1 \right)}{b_{min}-b_{max}}\ \text{when}\ v_1 < 2v_2,
            \label{eq:lower_bound_v1v2}
        \end{equation}
        then $b_1 \geq b_2 $, i.e., the relation between the two agents is cooperative, otherwise is competitive.
        \label{theorem}
\end{theorem}

The proof could be found in Appendix \ref{app:proof}. The above theorem indicates that when $\tau$ is larger than a threshold, agent 2 would prefer to cooperate; otherwise, she would behave competitively.
Note that if the impression value $v_1$ is sufficiently large, i.e., $v_1 \geq 2v_2$, then agent 2 would always cooperate\footnote{In practice, even when $v_1 \geq 2v_2$, agent 2 may not cooperate when $\tau \rightarrow 0$, which is slightly different from the derived result in theorem. This is because when $\tau \rightarrow 0$, softmax function would consistently output $\alpha_1=\alpha_2=0.5$ if $b_1=b_2$, but $\alpha_1=1$ and $\alpha_2=0$ if $b_1 > b_2$. However, in reality, their bids are hardly to be exactly the same, enabling our TRCA to establish a fully competitive case by setting $\tau \rightarrow 0$.
}, regardless of the value of $\tau$. We can safely conclude that a mixed cooperative-competitive relation among agents arises from setting $\tau > 0$, and the relation tends to be more cooperative by setting a large $\tau$, while more competitive by setting a small $\tau$. The hyper-parameter $\tau$ works as a proxy, through which the trade-off between competition and cooperation can be carefully controlled.

\subsection{Improving Revenue with Bar Agents}
\label{sec:bar_agent}

Although the cooperative approach contributes to a better social welfare, it leads to the aforementioned collusion behaviors which hurts the  platform's revenue. In this subsection, we introduce bar agents with different versions to avoid this.

The simplest way to improve the revenue is to set a \textbf{fixed bidding bar} $\bar{b}$. If the auto-bidding agent's bid satisfies $b_i \geq \bar{b}$, the assigned reward for her is set to $r^{\text{train}}_i = r^{\text{TRCA}}_i$, otherwise $r^{\text{train}}_i = 0$. However, the fixed bidding bar needs to be tuned elaborately to obtain satisfactory performance. A large value may badly reduce the advertisers' utilities, while an excessively small one may not effectively boost the revenue.


A more advanced method is setting an \textbf{adaptive bidding bar} for each impression opportunity to optimize the revenue. One can introduce a bar agent implemented by reinforcement learning (RL) to achieve this, with bidding bar as bar agent's action and available information of the impression opportunity as observation.
However, the reward function for the bar agent is not trivial to define. Simply defining the payment of each auction as the reward for bar agent would lead to an extremely large bidding bar that may reduce social welfare. Besides, this method also ignores the fact that different auto-bidding agents usually have different budgets and values when competing for the same impression opportunity, implying that the same bidding bar may not be a good choice.

Based on the above analysis, we propose personalized bidding bar agents to improve the revenue. 
In Figure \ref{fig.model}, MAAB introduces \textbf{multiple bar agents} $\{\bar{\pi}_i\}_{i=1}^n$, with each bar agent $\bar{\pi}_i$ aiming at setting a personalized  bar $\bar{b}_i$ for the corresponding auto-bidding agent $\pi_i$. Each bar agent shares the same observation as its corresponding auto-bidding agent. At each timestep, bar agents and the auto-bidding agents give their bidding bars $\{\bar{b}_i\}_{i=1}^n$ and bids $\{b_i\}_{i=1}^n$, respectively. But only $\{b_i\}_{i=1}^n$ are submitted to the auction. Then the auction environment returns the payment $p$ and the rewards $\{r_i\}_{i=1}^n$. The rewards $\{r_i\}_{i=1}^n$ are re-assigned by TRCA, obtaining $\{r_i^{\text{TRCA}}\}_{i=1}^n$.
However, using the payment $p$ as the reward for bar agents may lead to extremely large bidding bars according to the previous discussion. To circumvent this, we propose a reward scheme called \textbf{bar gate}. The bar gate outputs a scalar $z_i = z(\bar{b}_i, b_i) \in \{0, 1\}$ for each pair of auto-bidding agent and bar agent, indicating whether each auto-bidding agent's bid exceeds the bidding bar.
The rule underlying the bar gate is given by  
\begin{equation}
    \label{eq:bar_gate}
    z(b_i, \bar{b}_i)=\left\{\begin{array}{ll}
        1 \  & \text { if } b_i \geq \bar{b}_i, \\
        0 \  & \text { otherwise }.
    \end{array}\right.
\end{equation}
With the bar gate, the rewards for optimizing $\pi_i$ and $\bar{\pi}_i$ are $r^{\text{train}}_i = z_i \times r_i^{\text{TRCA}}$ and $\bar{r}^{\text{train}}_i=z_i \times p$, respectively. It is worth to note that bar agents are introduced during the training phase but removed during the {execution} time.

The bar agents and auto-bidding agents are trained simultaneously via an adversarial manner. The training procedure for $\bar{\pi}_i$ is to optimize the bidding bar $\bar{b}_i$, maximizing the revenue, while the auto-bidding agent $\pi_i$ wants to lower her bid to reserve more budget due to the cooperation relation. The bar gate connects these two different goals by enforcing the bar agent's bidding bar to be a maximum lower bound of auto-bidding agents' bid. 


The proposed multiple bar agents and the reward scheme underlying the bar gate allow us to improve revenue by raising the auto-bidding agents' bids to an appropriate level. It is worth to note that the bidding bar is similar to the reserve prices \cite{ostrovsky2011reserve,mohri2014learning}, but differs in that the bidding bar requires no modification to the GSP mechanism \cite{edelman2007internet} as it is only introduced during the training phase.

\subsection{Modeling Large-Scale Multi-Agent System}
\label{sec:mean_agent}

\begin{figure}[t]
    \centering
    \includegraphics[width=0.97\linewidth]{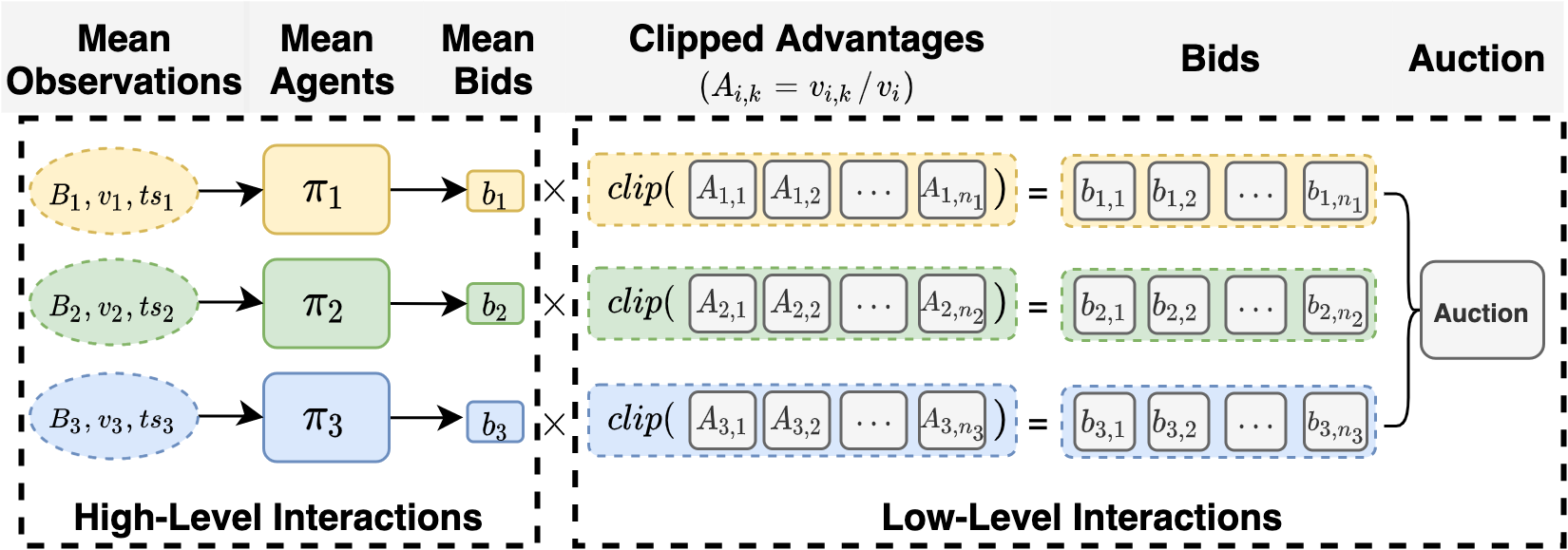}
    \caption{Modeling with mean agent approach.}
    \label{fig.mean_agent}
\end{figure}

In practical system, there may be millions of advertisers competing for billions of impressions every day, making it hard to train so many agents simultaneously by formulating each fine-grained advertiser as an agent, due to both the sparsity of the rewards and the limited resources (e.g., computational resources, time). We introduce our mean agent approach, which aims at providing a feasible and general method for the training of the large-scale multi-agent system for auto-bidding.

One can take a high-level perspective such as \emph{grouping} \cite{jin2018real}. 
By grouping advertisers by a certain principle, the rewards are no longer sparse at the group-level, and we only need to simultaneously train a separate policy for each high-level group instead of for each fine-grained advertiser. Our proposed mean agent first group advertisers by the \emph{objective}, as the main difference between advertisers lies in their optimizing objectives. Other principles may be possible depending on the specific situation. The grouping by objective results in a set of groups $\{G_1, G_2, \cdots, G_n\}$, where $G_i$ is a set containing all the advertisers belonging to this group. 
However, it is not a trivial work to train such a policy $\pi_i$ for group $G_i$ that can be properly shared by all advertisers $k\in G_i$ to generate their bids. This is due to the following two reasons: 1) the $Q$-learning updating rule requires calculating the maximum action value for the next state, but the next state on the group-level is not clear; 2) advertisers usually have different budgets and values when competing for the same impression opportunity, making the policy sharing across different advertisers within the same group difficult.

To handle the above issues, our mean agent approach considers the mean effects from a high-level perspective. The diagram of our mean agent approach is shown in Figure \ref{fig.mean_agent}. The main idea is that, for each group, we train a mean policy $\pi_i$ that calculates the mean bid based on the mean value and budget, and let each advertiser within the group derive her bid based on her value's \emph{advantage} over the mean value. 

Before giving the full details of our modeling, we first introduce some notations. We consider a timestep as a period of time (e.g., 15 minutes), during which a number of impression opportunities arriving sequentially. We denote $e \in E_t$ as an impression opportunity, where $E_t$ is a set containing all opportunities within timestep $t$. The  $v_{i,k}^e$ is the value of advertiser $k \in G_i$ for impression opportunity $e$. The $x^e_k \in \{0, 1\}$ indicates whether advertiser $k$ wins the impression opportunity $e$. Advertiser $k$ wins the impression if she has the maximum effective cost-per-mille (eCPM) ranking score $pCTR^e_k \times b_{k}^e$. The specific modeling under the framework of Markov Games \cite{littman1994markov} is explained as follows.

\noindent \textbf{Observation space}: the observation of mean agent $i$ at timestep $t$ is defined as $o_i^t=(B_i^t, v_i^t, ts_i^t)$. The $B_i^t$ is  mean agent's remaining budget at this timestep, and the initial budget of mean agent is set to $B^1_i= \frac{1}{|G_i|}\sum_{k\in G_i} B_k$. The mean value $v_i^t = \mathbb{E}[v_{i,k}^e] \approx \frac{1}{|G_i|\times |E_t|}\sum_{e \in E_t, k \in G_i} v_{i,k}^e$. The $ts_i^t$ is  timesteps left along the episode.

\noindent \textbf{ Action space}: each mean agent takes mean bid $b_i^t \in \mathcal{A}_i$ as its action. The bid of advertiser $k \in G_i$ for impression opportunity $e$ is $b_{i,k}^e= b_i^t \times clip(A_{i,k}^e)$, where $A_{i,k}^e=v_{i,k}^e/v_i$ is defined as the \emph{advantage} 
of $v_{i,k}^e$ over the mean value. The advantages are clipped for preventing advertisers having excessively large values.

\noindent \textbf{Reward function}: reward function is defined at the group-level \cite{wu2018budget} due to the large number of impressions and advertisers. Mean agent $i$'s reward is defined as $r_i^t = \frac{1}{|E_t|}\sum_{e \in E_t, k \in G_i} v_{i,k}^e \times x_{k}^{e}$.

\noindent \textbf{Transition function}: The impression-level payment for winning impression $e$ is $p^e = pCTR^{e}_{j} \times b^{e}_{j}$ in expectation, where $j$ is the index of the next ranked advertiser according to maximum eCPM ranking score. The payment for the mean agent is also defined at the group-level: $p_i^t = \sum_{e \in E_t, k \in G_i} p^e \times x^{e}_k$. Thus each mean agent's next observation $o_i^{t+1}=(B_i^t - p_i^t, v_i^{t+1}, ts_i^t-1)$. Mean agent's action is forced to  $0$ when its budget is below $0$.

Note that, for online deployment, the mean policy for a group is shared by all advertisers within the group to generate their bids separately without calculating advantages. The mean agent modeling can also be easily generalized to our proposed TRCA and bar agents approaches by replacing the bids and bidding bars with the corresponding mean ones.



\section{Experiments}

We evaluate our method in an offline industrial dataset and perform an online A/B test on the Alibaba e-commerce advertising platform.

\subsection{Offline Dataset Simulation}

We perform the offline evaluation on a real offline dataset from Alibaba's advertising system where on average 10 billion auctions are covered per day. See Appendix \ref{app:offline_simulation} for the environmental setup and implementation details.

\noindent \textbf{Offline Dataset.} The offline dataset is extracted from the six-hour search auction log of November 18, 2020, including 705,140 impression opportunities. In each impression opportunity, nearly 400 ads are recalled to compete for displaying. Each recalled ad is associated with an advertiser's id, timestamp, objective, impression value, manually set bid, etc. The objectives are mainly of three types -- the number of clicks (pCTR), conversions (pCVR $\times$ pCTR), and add-to-carts (pCART $\times$ pCTR), which corresponds to three groups (CLICK, CONV, CART) in our experiments. A similar preprocessing as \cite{jin2018real} is adopted: we randomly sample 1/150 of the logged impressions from the whole dataset for training, and the dataset for testing is sampled in the same way.

\noindent \textbf{Evaluation Metrics.} We adopt the following evaluation metrics: 1) social welfare, which is calculated by adding up the normalized total values won by three groups 
; 2) platform's revenue, which is defined as the cumulative payments along the episode. The payment at each timestep is calculated by GSP-based expected Cost-Per-Click (CPC). The evaluation procedure is provided in Appendix \ref{app:eval}.

\noindent \textbf{Budget Constraints.} For the offline experiments, we first calculate the total payment of an episode by enforcing all mean agents using their maximum mean bids, and accumulate the total payment $P = \sum_t \sum_{e \in E_t} p^e$. Then the budget of advertiser $k \in G_i$ and mean agent $i$ are set as a fraction of the episode's total payment: $B_i=B_{k}= P \times B_0 \times r[i]$. We run the evaluation with $B_0=1/4$ and $r=[1,1,1], [1.5, 0.5, 1]$.

\noindent \textbf{Compared Methods.} With the same settings, the following methods are selected for comparisons: 1) \textbf{MSB} uses advertisers' manually set bids for bidding. 2) \textbf{DQN-S} can be seen as the single-agent version of IL. Specifically, we train a mean policy for each group using DQN \cite{mnih2015human} by assuming other advertisers use their manually set bids. However, all mean agents use their trained policies for auction during testing rather than considering other advertisers' bids as the manually set ones. 3) \textbf{CM-IL} \cite{tan1993multi}. 4) \textbf{CO-IL}. 5) \textbf{MAAB} is our proposed method, which is the mixed cooperation-competitive IL augmented with bar agents.


\noindent \textbf{Experimental Results.} We test the performance under the following two settings: 1) $B_0=1/4$ and $r=[1,1,1]$, in which case all auto-bidding agents have equal budget constraints; 2) $B_0=1/4$ and  $r=[1.5,1,0.5]$, which is an unbalanced budget setting. The main results are shown in Table \ref{tab:offline_1}. The reported performance is averaged over 3 independent runs after training for 3.5 million timesteps.

\begin{table}
    \caption{Mean and standard deviation of different groups' values (CLICK, CONV, CART), platform's revenue, and social welfare in offline dataset simulation. }
    \label{tab:offline_1}
    \scalebox{0.85}{%
        \begin{tabular}{cccccc}
            \toprule
            Setting 1     & CLICK                 & CONV                  & CART                  & Revenue               & Social Welfare         \\
            \midrule
            MSB                           & 24.7$\pm$0            & 21.8$\pm$0            & 18.0$\pm$0            & 16.9$\pm$0            & 64.5$\pm$0             \\
            DQN-S                         & \textbf{29.3}$\pm$2.7 & 35.8$\pm$5.1          & \textbf{36.0}$\pm$2.3 & 68.3$\pm$6.7          & 101.0$\pm$2.5          \\
            CM-IL                        & 27.8$\pm$0.9          & 41.3$\pm$0.7          & 35.0$\pm$0.8          & \textbf{86.8}$\pm$1.2 & 104.1$\pm$0.8          \\
            CO-IL                        & 27.3$\pm$1.5          & 41.3$\pm$2.0          & \textbf{35.6}$\pm$1.7 & 66.9$\pm$10.2         & 104.3$\pm$2.3          \\
            MAAB                          & 28.0$\pm$0.8          & \textbf{41.8}$\pm$1.3 & 35.5$\pm$1.4          & 80.6$\pm$3.2          & \textbf{105.3}$\pm$1.3 \\
            \toprule
            Setting 2 & CLICK                 & CONV                  & CART                  & Revenue               & Social Welfare         \\
            \midrule
            MSB                           & 24.7$\pm$ 0           & 21.8$\pm$0            & 18.0$\pm$0            & 16.9$\pm$0            & 64.5$\pm$0             \\
            DQN-S                         & \textbf{37.1}$\pm$2.1 & 25.1$\pm$3.1          & 34.8$\pm$1.8          & 75.5$\pm$5.2          & 97.0$\pm$2.3           \\
            CM-IL                        & 35.3$\pm$1.3          & 29.2$\pm$1.0          & 35.1$\pm$0.8          & \textbf{85.0}$\pm$2.9 & 99.6$\pm$0.6           \\
            CO-IL                        & 30.7$\pm$3.3          & \textbf{35.3}$\pm$3.2 & 37.0$\pm$2.7          & 52.9$\pm$12.4         & 103.0$\pm$2.4          \\
            MAAB                          & 31.3$\pm$1.2          & 33.5$\pm$1.7          & \textbf{38.6}$\pm$1.4 & 69.0$\pm$3.6          & \textbf{103.4}$\pm$0.7 \\
            \bottomrule
        \end{tabular}
    }
\end{table}

We find that traditional way of manually setting bids (MSB) fails to achieve good performance -- social welfare is 64.5 and platform's revenue is 16.9, which is consistently the worst among all methods. By comparison, DQN-S is superior in terms of the three groups' values (29.3, 35.8, 36.0), social welfare (101.0) and platform's revenue (68.3), due to the budget spending controlled by RL.

However, the performance of DQN-S is still limited by the unreliable assumption that other agents' bids are fixed. This assumption can be further removed by adopting a multi-agent learning paradigm. The benefits of multi-agent learning can be demonstrated by CM-IL's superior performance over DQN-S both in terms of social welfare (e.g., 104.1 > 101.0 in setting 1) and platform's revenue (e.g., 86.8 > 68.3 in setting 1). CM-IL learns auto-bidding policies simultaneously, making it possible to model the interactions between agents explicitly. However, the competitive relation may not help achieve better social welfare, which can be seen by comparing CM-IL with CO-IL. CO-IL models the relations between auto-bidding agents in a cooperative way and is slightly better than CM-IL in social welfare (104.3 > 104.1 in setting 1 and 103.0 > 99.6 in setting 2), however, at the cost of reducing the platform's revenue (66.9 < 86.8 in setting 1 and 52.9 < 85.0 in setting 2).

In between these two extremes, MAAB uses TRCA and models the relation between agents in the MCC way, which achieves a better equilibrium between social welfare and the revenue. As shown in Table \ref{tab:offline_1}, MAAB achieves better social welfare compared with CM-IL (105.3 > 104.1 in setting 1 and 103.4 > 99.6 in setting 2) and significantly outperforms CO-IL in terms of the revenue (80.6 > 66.9 in setting 1 and 69.0 > 52.9 in setting 2).


\subsection{Online Experiments}

To further verify our MAAB method's validity, we have launched it to the online real production environment (Taobao display advertising system) and compared it with the baseline method CM-IL.
The deployment has the following details: 1) $Q$ networks in both methods are trained on a Tensorflow-based distributed training framework and updated every 2 hours. 2) More features (such as pCTR, pCVR, etc.) are encoded into the observation in both methods for better modeling the dynamic environment. 3) As mentioned in subsection \ref{sec:mean_agent}, during online deployment, the bid for each ad is generated by feeding each ad's encoded observation into the corresponding mean policy.

The main results of the online A/B test with 1\% of whole production traffic from January 30, 2021 to February 3, 2021 are shown in Table \ref{tab:online_1}. Since real online data needs to be kept confidential, we normalized the baseline method's performance to 100. Similar to the results of offline experiments, compared with CM-IL, MAAB achieves better social welfare, but with a lower revenue (96.1 < 100).

\begin{table}[t]
    \caption{Mean of different groups' values (CLICK, CONV, CART), platform's revenue and social welfare in the online production environment.}
    \label{tab:online_1}
    \scalebox{0.9}{%
        \begin{tabular}{cccccc}
            \toprule
            \small{}     & CLICK                 & CONV                  & CART                  & Revenue               & Social Welfare         \\
            \midrule

            CM-IL                        & 31.4          & 48.2          & 20.4          & \textbf{100.0} & 100.0          \\

            MAAB                          & \textbf{32.9}          & \textbf{50.3} & \textbf{21.4} & 96.1    & \textbf{104.6} \\
            \bottomrule
        \end{tabular}
    }
\end{table}

\subsection{Ablation Study}

This section investigates the effectiveness of TRCA and the necessity of bar agents for improving platform's revenue.

\subsubsection{Effectiveness of TRCA}

\begin{figure}[htbp]
    \centering
    \subfigure[Social welfare]{
        \label{fig.ablation_mixed_sw}
        \includegraphics[width=0.23\textwidth]{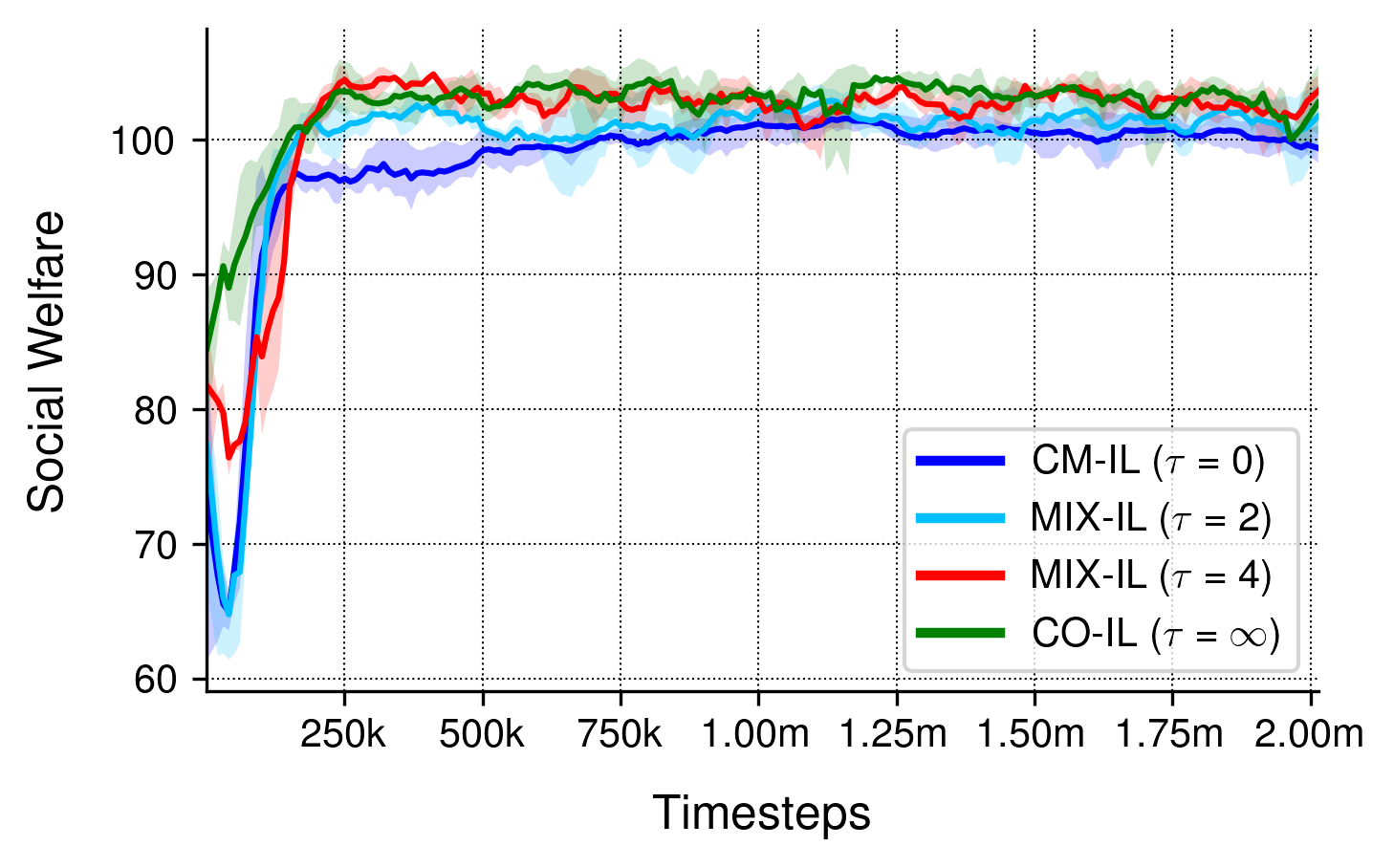}}
    \subfigure[Platform's revenue]{
        \label{fig.ablation_mixed_pr}
        \includegraphics[width=0.23\textwidth]{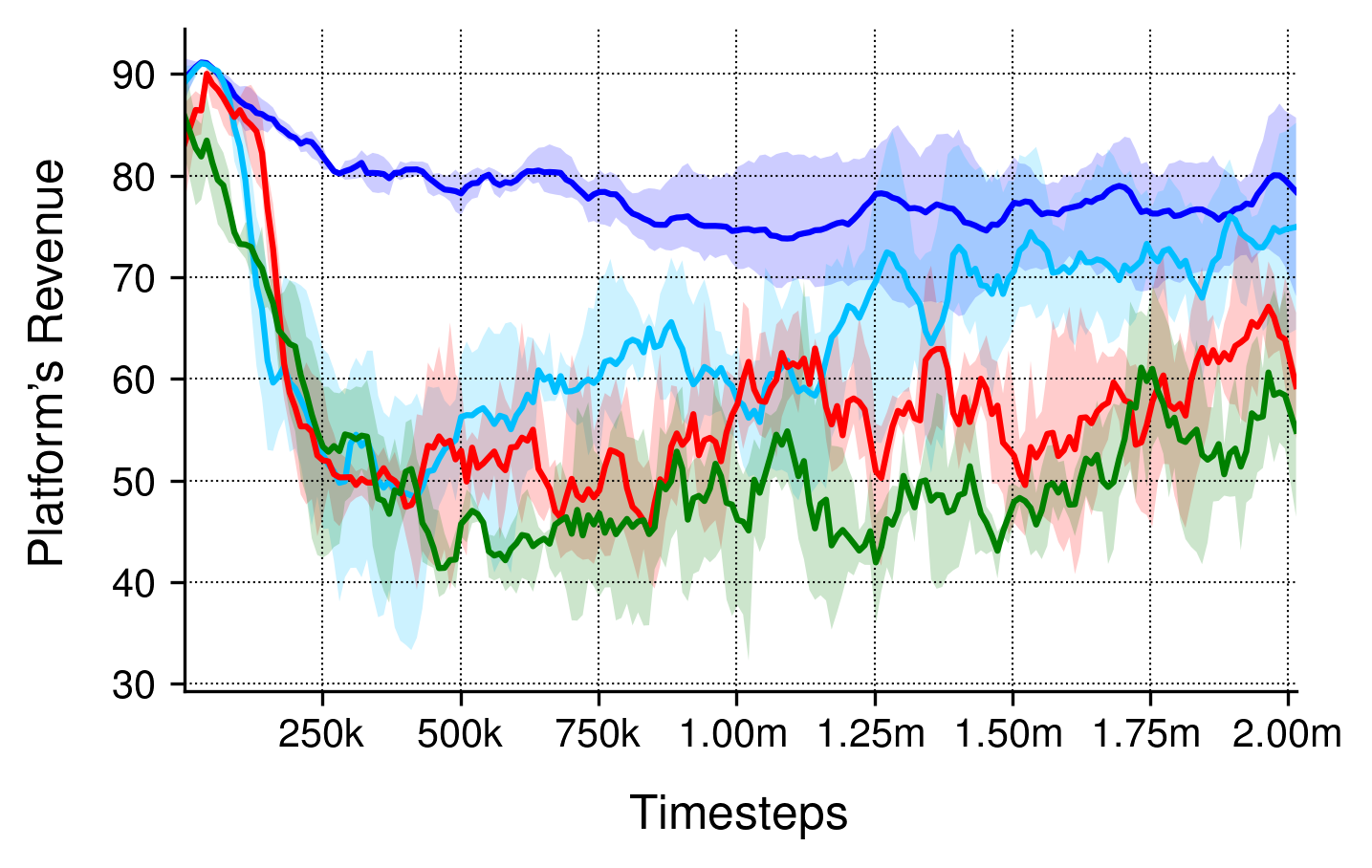}}
    \caption{Social welfare and platform's revenue for methods with different parameter $\tau$. The mean and 95\% confidence interval are shown across 3 independent runs.}
    \label{fig.ablation_mixed}
\end{figure}

To investigate the effectiveness of the proposed TRCA for modeling the cooperative and competitive relation, we remove the bar agents in MAAB and call the resulting method MIX-IL. Then we adjust the parameter $\tau$ in MIX-IL and perform experiments with our offline dataset, where a larger $\tau$ corresponds to more cooperative whereas a smaller $\tau$ corresponds to more competitive. Especially, $\tau=0$ is equivalent to CM-IL; $\tau=\infty$ is equivalent to CO-IL. We set $B_0=1/4$ and $r=[1.5, 0.5, 1]$. 

As shown in Figure \ref{fig.ablation_mixed_sw}, after training for 2 million timesteps,
CO-IL and MIX-IL achieve nearly the same social welfare (101.2, 103.1 and 101.9 for $\tau=2, 4$ and $\infty$ respectively), and they consistently outperform CM-IL (99.6), which suggests cooperation can help achieve better social welfare. However, Figure \ref{fig.ablation_mixed_pr} illustrates the side effect of cooperation. A fully cooperative approach would significantly reduce the platform's revenue (57.5), while a competitive approach guarantees a higher platform's revenue (79.3) due to the competition. In between these two extremes, MIX-IL ($\tau=2$) seemingly achieves a better equilibrium that guarantees good social welfare (101.2) comparable with the fully cooperative one (101.8) and improves platform's revenue (74.8 > 57.5). In summary, to achieve a win-win situation for advertisers and the platform, it seems necessary to balance the competition and cooperation in multi-agent auto-bidding problem, and $\tau$ could serve as a convenient tool to achieve such a balance.

\subsubsection{Influence of Bar Agents}

We investigate necessity of bar agents for improving platform's revenue and the effectiveness of adaptive bidding bars. To this end, we compare MAAB against the following two methods in offline dataset simulation: 1) MIX-IL: this method is derived from MAAB by removing bar agents; 2) MAAB-fix: we fix bar agents' actions to be a heuristic value ($\bar{b}=1$ and $\bar{b}=4$) in MAAB, which can be seen as the \emph{fixed bidding bar} introduced in subsection \ref{sec:bar_agent}. All methods are trained for 2 million timesteps, and we set $B_0=1/2$ and $r=[1.5, 0.5, 1]$. Here we use a larger $B_0$, which corresponds to more budget in total and can help better illustrate the bar agents' effectiveness in improving the platform's revenue.

As shown in Table \ref{table:ablation_vrl}, by comparing MIX-IL with MAAB-fix ($\bar{b}=1$ or $\bar{b}=4$), we can see that the idea of setting a bidding bar could improve the platform's revenue. For example, platform's revenue increases from 99.6 to 114.3 by using bidding bar $\bar{b}=1$ in MIX-IL, and further increases to 164.9 by raising the bidding bar to $\bar{b}=4$. However, a large bidding bar (i.e., $\bar{b}=4$) would hurt the social welfare, while a small bidding bar (i.e., $\bar{b}=1$) still leaves room for further improvement of platform's revenue. The benefit of the adaptive bidding bar can be seen by comparing MAAB with MAAB-fix ($\bar{b}=1$). MAAB achieves comparable social welfare (103.9 $\approx 104.5$) with MAAB-fix ($\bar{b}$=1) and further improves platform's revenue (134.6 > 114.3). Our adaptive bidding bar and the bar gate enable MAAB to improve the revenue without reducing the social welfare, which presents a clear benefit over the fixed one. See Appendix \ref{app:bar_agents_more} for additional analysis of the bar agents.

\begin{table}[t]
    \caption{Mean and standard deviation of social welfare and platform's revenue.}
    \label{table:ablation_vrl}
    \scalebox{0.95}{%
    \begin{tabular}{c|cc}
        \multicolumn{2}{c}{}                                                     \\
        \toprule
                               & Social Welfare         & Platform's Revenue     \\
        \hline
        MIX-IL                & 104.0$\pm$3.3          & 99.6$\pm$18.2          \\
        MAAB-fix ($\bar{b}=1$) & \textbf{104.5}$\pm$1.4 & 114.3$\pm$17.1         \\
        MAAB-fix ($\bar{b}=4$) & 99.3$\pm$0.4           & \textbf{164.9}$\pm$1.3 \\
        MAAB                   & 103.9$\pm$1.2          & 134.6$\pm$8.2          \\
        \hline
    \end{tabular}
    }
\end{table}

\section{Related Work}

\noindent \textbf{Multi-Agent Reinforcement Learning.}
Recent works have explored RL beyond single-agent scenarios and have considered multi-agent case \cite{tan1993multi,lowe2017multi,rashid2018qmix,wen2020smix}. In MARL literature, some works focus on learning cooperation \cite{rashid2018qmix,sunehag2018value,tan1993multi}, in which a team of agents takes actions to achieve a common goal. One of the challenges in learning cooperation is the credit assignment problem \cite{foerster2018counterfactual,sunehag2018value,rashid2018qmix}, i.e., the agents should deduce their contributions to the global rewards \cite{hernandez2019survey}. Our TRCA instead considers the credit assignment for establishing a MCC relation among agents, which differs from previous works in terms of both motivation and the problem settings. 



\noindent \textbf{Bid Optimization.} Bidding optimization is one of the key problems in real-time bidding. Several lines of work consult RL to learn the bidding strategies under constraints \cite{cai2017real,jin2018real,wu2018budget}. 
However, these works adopt a single-agent setting -- they optimize the bid strategy under the assumption that the market price is stationary. A few studies also consider bid optimization from the multi-agent perspective. \citet{jin2018real} formulate bid optimization with cooperative MARL and use a clustering method for dealing with the huge number of advertisers. However, the cooperative game may result in collusion behaviors that reduce the platform's revenue. Concurrently to our work, \citet{guan2021multi} try to avoid the collusion behaviors by introducing an extra revenue constraint. We instead avoid collusion by modeling the MCC relation \cite{tampuu2017multiagent} among agents and design bar agents to improve the revenue further in an adversarial manner. To optimize revenue, most works focus on the reserve prices \cite{mohri2014learning,thompson2013revenue} and designing revenue-maximizing auctions \cite{dutting2019optimal,neuralauction}, which is a different setting from ours as we consider the joint optimization of the advertiser's bidding strategies and platform's revenue without changing the GSP mechanism.


\section{Conclusions}


In this work, we proposed a MARL framework, called MAAB, for the auto-bidding in online advertising through three contributions: 1) proposing TRCA for establishing a mixed cooperation and competition relation among agents, 2) using bar agents and a reward scheme, called bar gate, for improving the platform's revenue with an adversarial training manner, 3) proposing a mean agent approach for the deployment of our methods on the large-scale advertising platform. Extensive offline and online experiments demonstrate the effectiveness of MAAB in achieving social welfare and guaranteeing the platform's revenue. Our future work includes extending TRCA by dynamically tuning the temperature parameter based on real-time information. Besides, further investigation on the reward scheme underlying the bar gate is needed for the fast convergence of bar agents.

\bibliographystyle{ACM-Reference-Format}
\bibliography{sample-base}

\newpage
\appendix

\section{Algorithm}
\label{app:algo}

\begin{algorithm}
\small
Set the empty replay buffer $\mathcal{D}$ to capacity $N_\mathcal{D}$, $step = 0$, training batch size $b$\;

Initialize Q-networks $Q_i$ and $\bar{Q}_i$ with random parameters $\theta_i$ and $\hat{\theta}_i$, and the corresponding target networks with parameters $\hat{\theta}_i \leftarrow \theta_i$ and $\hat{\bar{\theta}}_i \leftarrow \bar{\theta}_i$, for $i=1, \cdots, n$\;

\While{$step < max\_step$}{
        \For{$t=1$ to $T$}{
            \For{each agent i}{
                 Obtain the observation $o_i$ for agent $i$ and bar agent $i$\;
                Select $b_i$ according to $\epsilon$-greedy policy w.r.t  $Q_i$\;
                Select $\bar{b}_i$ according to $\epsilon$-greedy policy w.r.t  $\bar{Q}_i$\;
            }
            Submit $\{b_1, \cdots, b_n\}$ to the auction environment\;
            Get rewards $\{r_1, \cdots, r_n\}$ and the payment $p$\;
            Calculate $r_i^{\text{TRCA}}$ according to Eq. 4\;
            Get $z_i=z(b_i, \bar{b}_i)$ according to Eq. 6\;
            Calculate $r^{\text{train}}_i = z_i \times r_i^{\text{TRCA}}$ and $\bar{r}^{\text{train}}_i=z_i \times p$\;
            $step = step + 1$
        }
        Store the episode in $\mathcal{D}$, replacing the oldest episode if $|D| > N_\mathcal{D}$\;
    
    Sample a batch of b episodes $\sim \text{Uniform} (\mathcal{D})$\;
    \For{each agent i}{
        Set $y_i = r_i^{\text{train}} + \gamma \max_{b'_i} Q_i(o'_i, b'_i; \hat{\theta})$\;
        Set $\bar{y}_i = \bar{r}_i^{\text{train}} + \gamma \max_{\bar{b}'_i} \bar{Q}_i(o'_i, \bar{b}'_i; \hat{\bar{\theta}})$\;
        
        Update $\theta_i$ by minimizing $\sum_{b,t} \left[\left(y_{i}-Q_{i}\left(o_{i}, b_{i} ; \theta_{i}\right)\right)^{2}\right]$\;
    
        Update $\bar{\theta}_i$ by minimizing $\sum_{b,t} \left[\left(\bar{y}_{i}-\bar{Q}_{i}\left(o_{i}, \bar{b}_{i} ; \bar{\theta}_{i}\right)\right)^{2}\right]$ twice\;
        
        Replace target parameters $\hat{\theta_i} \leftarrow \theta_i$ every C episodes\;
        Replace target parameters $\hat{\bar{\theta_i}} \leftarrow \bar{\theta_i}$ every C episodes\;
    }

}

\caption{Training Procedure for MAAB}\label{algo:2}
\end{algorithm}

\section{Proof of Theorem \ref{theorem}}
\label{app:proof}

\newtheorem*{remark}{Theorem 4.1}

\begin{remark}
        Consider a two-agent bidding case with impression values satisfying $v_1 > v_2$ in one round auction. Let $b_1, b_2 \in [b_{min}, b_{max}]$ denote two agents' bids, where $b_{min}$ and $b_{max}$ are the minimum and maximum allowable bids respectively. If $ v_1 \geq 2 v_2$ or

        \begin{equation}
			\tau \geq \frac{\log \left( 2v_2/v_1-1 \right)}{b_{min}-b_{max}}\ \text{when}\ v_1 < 2v_2, \notag
        \end{equation}
        then $b_1 \geq b_2 $, i.e., the relation between the two agents is cooperative, otherwise is competitive.
\end{remark}

\begin{proof}
	Let $\mathbf{b}=(b_1,b_2)$.
	For agent 2, the objective is
    \begin{equation}
    \max_{b_2}\ \alpha_2 (\mathbf{b}) \cdot \big[\mathbbm{1}(b_1 \geq b_2) \cdot v_1 + \big(1-\mathbbm{1}(b_1 \geq b_2)\big) \cdot v_2\big],
    \label{eq:agent2obj}
    \end{equation}
	where $\alpha_2 ( \mathbf{b} ) = \frac{\text{exp}\{b_2/\tau\}}{\sum_{j=1}^2 \text{exp}\{b_j/\tau\}}$ and $\mathbbm{1}(\cdot)$ is the indicator function. We define the the solution set of case L and H: $L = \{(b_1, b_2)|b_{max} \geq b_1 \geq b_2 \geq  b_{min}\}$, $H = \{(b_1, b_2)|b_{max} \geq b_2 > b_1 \geq b_{min}\}$. 
	Our goal is to find the condition that the optimal solution of Eq. \ref{eq:agent2obj} always lies in the case L, which can be transformed into the following target:
	\begin{equation}
	    \max_{\mathbf{b}^L \in L } \alpha_2 ( \mathbf{b}^L ) \cdot v_1 
		\geq 
		\max_{\mathbf{b}^H \in H }\alpha_2 ( \mathbf{b}^H ) \cdot v_2.
		\label{eq:proof_target}
	\end{equation}
Eq. \ref{eq:proof_target} can be rewritten as
\begin{align}
g(\tau)
&=\frac{\max_{\mathbf{b}^L\in L} \alpha_2(\mathbf{b}^L) \cdot v_1}{\max_{\mathbf{b}^H\in H} \alpha_2(\mathbf{b}^H) \cdot v_2} \notag
\\
&=\frac{v_1}{2v_2}\frac{1}{\max_{\mathbf{b}^H\in H} \alpha_2(\mathbf{b}^H)} \notag
\\
&=\frac{v_1}{2v_2}\min_{\mathbf{b}^H\in H} \frac{\exp \{b^{H}_{1}/\tau \}+\exp \{b^{H}_{2}/\tau \}}{\exp \{b^{H}_{2}/\tau \}} \notag
\\
&=\frac{v_1}{2v_2}\min_{\mathbf{b}^H\in H} \left( \exp \{(b^{H}_{1}-b^{H}_{2})/\tau \}+1 \right)  \notag
\\
&=\frac{v_1}{2v_2}\left( \exp \{(b_{min}-b_{max})/\tau \}+1 \right) \geq 1 \label{eq_inequality_aaa}
\end{align}

Solving Eq. \ref{eq_inequality_aaa}, we obtain the value of $\tau$ for a cooperative relation between the two agents (i.e., $b_1 \geq b_2$) in Theorem \ref{theorem}.
\end{proof}

\section{Offline Dataset Simulation} 
\label{app:offline_simulation}

\subsection{Environmental Setup} 
\label{app:eval}

We adopt the modeling in subsection \ref{sec:mean_agent}. We consider each hour as an episode of length 60, with each minute in an hour as a timestep. In implementation, the individual reward for each mean agent is normalized for the sake of calculating social welfare: for mean agent $i$ at timestep $t$, we normalized the individual reward $r_i^t$ by dividing it by the maximum total value ($V^{\max}_i = \sum_{t, e \in E_t, k\in G_i} v_{i,k}^e$) that mean agent $i$ could achieve along the episode. 

\subsection{Evaluation Procedure}

The evaluation procedure is similar to the training process. We evaluate the performance by pausing training after every 10,000 timesteps and running 5 independent test episodes with bar agents removed and each agent performing greedy action selection in a decentralized way. The reported performance is also normalized for better comparisons: for calculating the performance of each group, we first calculate the total value obtained by each group ($V_i = \sum_{t, e \in E_t, k\in G_i} v_{i,k}^e \times x_{k}^e$) and divide it by the maximum total value $V^{\max}_i$ that this group could achieve along the episode. $V_i/V^{\max}_i$ is adopted as the performance of the group. Social welfare is therefore calculated as the sum of group performance.

\subsection{Implementation Details}
\label{app:imple_detail}

Each agent's $Q$ or $\bar{Q}$ network is a fully connected neural network with 3 hidden layers and 64 nodes for each layer. Each $Q$ or $\bar{Q}$ network takes as input the observation that includes the remaining budget, mean value and the timesteps left along the episode, and outputs the state-action $Q$ or $\bar{Q}$ value for each candidate action. Each agent selects actions from the interval $[0,5]$ that is discretized into 21 equally spaced values, according to its $Q$ or $\bar{Q}$ network.

During training, each agent explores the environment by following an $\epsilon$-greedy policy with linear $\epsilon$-annealing over 50k steps from 1.0 to 0.05. We uniformly sample 32 episodes from the replay buffer that contains the most recent 5000 episodes and perform a single training step on $Q$ networks after every episode. However, for faster convergence, $\bar{Q}$ networks are updated twice after every episode. The target networks for both $Q$ and $\bar{Q}$ are updated every 200 training episodes. 

To speed up the learning, mean agents and bar agents share the parameters of the $Q$ and $\bar{Q}$ network, respectively. Thus a one-hot encoding of the agent id is concatenated onto each agent's observations. All neural networks are trained using RMSprop with a learning rate $0.0005$. We set $\gamma=0.99$ and $\tau=4$. The advantages are clipped to the range [0, 3].






\subsection{Additional Analysis of the Bar Agents}
\label{app:bar_agents_more}

\begin{figure}[htbp]
    \centering
    \subfigure[Agents' bids for MAAB]{
        \label{fig.act_vil4_c01}
        \includegraphics[width=0.22\textwidth]{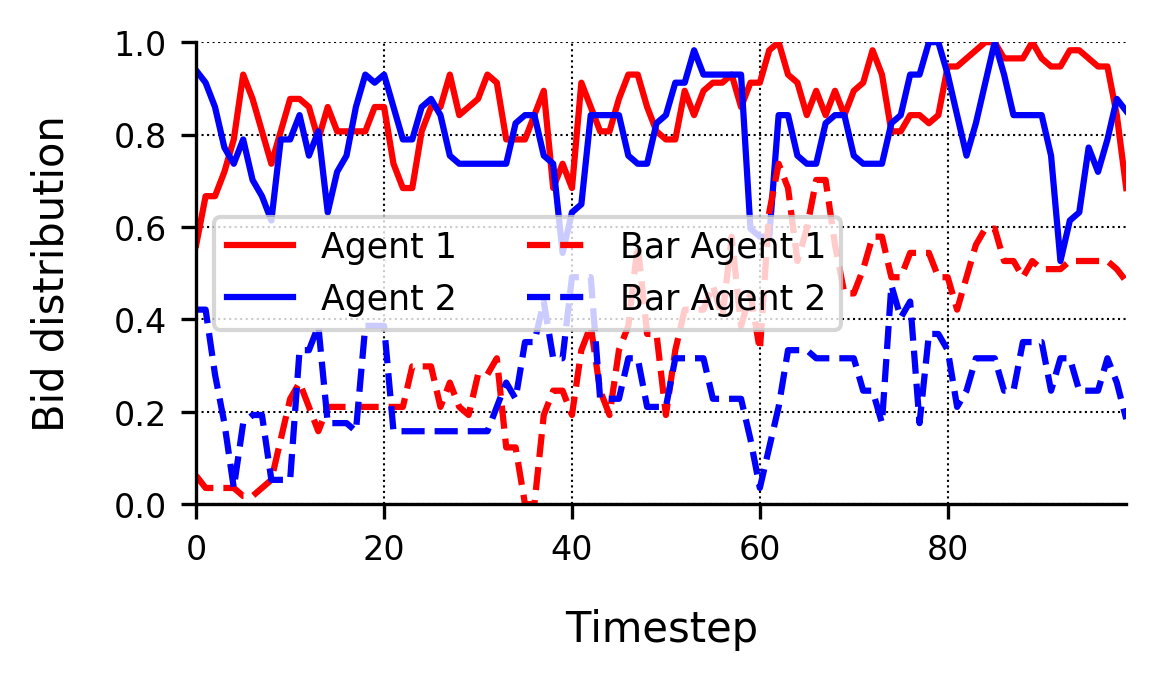}}
    \subfigure[Payments for three methods]{
        \label{fig.act_min}
        \includegraphics[width=0.22\textwidth]{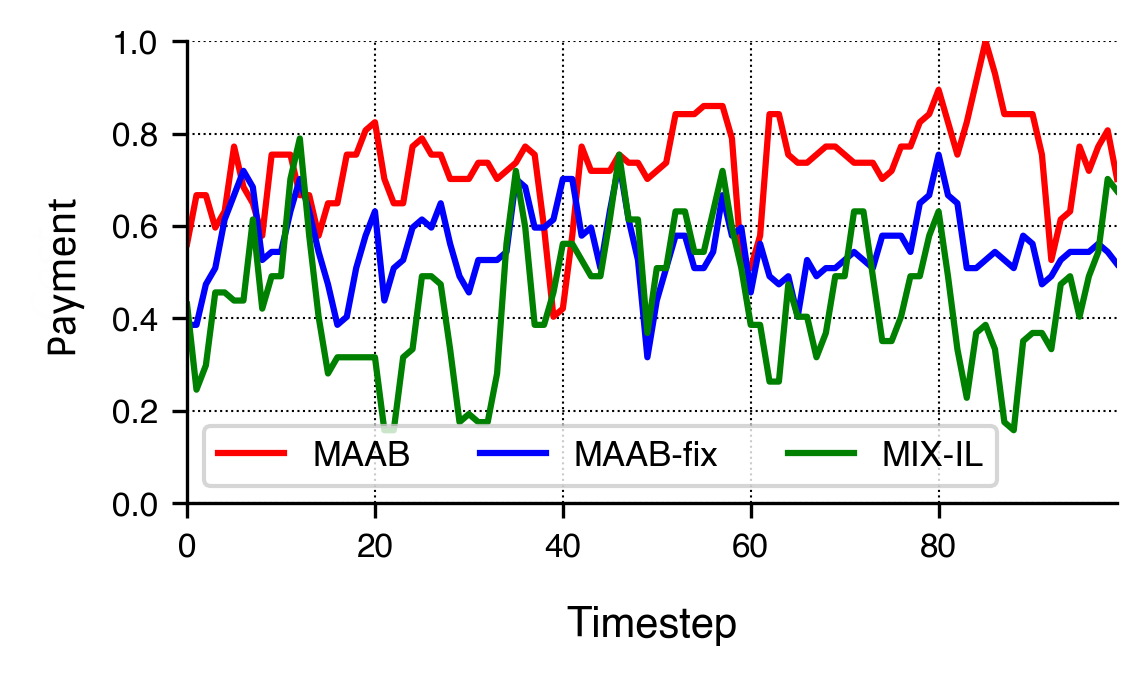}}
    \caption{Agents' bids for MAAB and the payments for three methods across each timestep of a selected episode. The curves are smoothed over a sliding window of size 3.}
    \label{fig.ablation_vrl}
\end{figure}

To further illustrate the effectiveness of bar agents for improving the platform's revenue, we select an episode and plot the bids of two agents across each timestep. The episode is selected in the two-agent bidding environment for the simplicity of illustration. We set $B_0=1$ and $r=0.5$. 

As shown in Figure \ref{fig.act_vil4_c01}, by comparing the solid line with the dotted line in the same color, we find that bar agents' bids are strictly lower than those of their counterparts. This demonstrates the effectiveness of bar gate as the violation of rule $b_i \geq \bar{b}_i$ in the bar gate prohibits both agents from receiving their rewards. Interestingly, our results also reveal that bar agents' bids have a similar pattern as their counterparts. For example, both agent 1 (blue line) and bar agent 1 (blue dotted line) have a bidding peak at timestep 19 and a bidding valley at timestep 60. For agent 2 (red line) and its counterpart (red dotted line), similar patterns can also be found at timestep 40 and 62. These patterns demonstrate that bar agents, along with the bar gate, can adaptively set the bidding bar by utilizing additional information such as impression value and budget, which presents a clear benefit over MAAB-fix that does not utilize extra information when determining the bidding bar.

In Figure \ref{fig.act_min}, we also plot the payments for MAAB, MAAB-fix ($\bar{b}=1$) and MIX-IL across each timestep. We find that approaches with bidding bar (MAAB, MAAB-fix) outperform the one without it (MIX-IL) in terms of the payments, which implies that the bidding bar could improve the revenue. Besides, payments of MAAB are higher than MAAB-fix, which suggests adaptively setting the bidding bar is superior to a static one.



\end{document}